\documentclass[11pt]{article}
\usepackage{etex}
\usepackage{amsmath}
\usepackage{fancyhdr}
\usepackage{amssymb}
\usepackage{amsthm}
\usepackage{graphicx}
\usepackage{varioref}
\usepackage{verbatim} 
\usepackage{multicol}
\usepackage{lmodern}
\usepackage{enumerate}
\usepackage[normalem]{ulem}
\usepackage{caption}
\usepackage{subcaption}
\usepackage[T1]{fontenc}
\usepackage[margin=1in]{geometry}
\usepackage{fancyhdr}
\usepackage{authblk}

\usepackage{mathrsfs}

\usepackage{url}
\usepackage{xspace}
\usepackage{thm-restate}

\usepackage{hyperref}
\hypersetup{
    bookmarksnumbered=true, 
    unicode=false, 
    pdfstartview={}, 
    pdftitle={}, 
    pdfauthor={}, 
    pdfsubject={}, 
    pdfcreator={}, 
    pdfproducer={}, 
    pdfkeywords={}, 
    pdfnewwindow=true, 
    colorlinks=true, 
    linkcolor=blue, 
    citecolor=blue, 
    filecolor=blue, 
    urlcolor=blue 
}

\newtheorem{theorem}{Theorem}

\newtheorem{definition}{Definition}
\newtheorem{lemma}{Lemma}
\newtheorem{remark}[theorem]{Remark}
\newtheorem{corollary}[theorem]{Corollary}

\theoremstyle{definition}

\theoremstyle{remark}

\usepackage{tikz} 

\input{Qcircuit}
\makeatother

\newcommand\QMA{{\sf{QMA}}}

\newcommand\PSPACE{{\sf{PSPACE}}}
\newcommand\FPSPACE{{\sf{FPSPACE}}}
\newcommand\PP{\sf{PP}}

\newcommand\NP{{\sf{NP}}}

\newcommand\BQP{{\sf{BQP}}}
\newcommand\PostBQP{{\sf{PostBQP}}}
\newcommand\PQP{{\sf{PQP}}}

\newcommand\QMAexp{{\sf{QMA}_{exp}}}
\newcommand\bddQMA[5]{{\left(#1,#2\right)}\textit{-bounded }\QMA_{#3}(#4,#5)}
\newcommand\PQPSPACE{{\sf{PQPSPACE}}}
\newcommand\revPSPACE{{\sf{revPSPACE}}}
\newcommand\succdet{\textit{Succinct Determinant Checking}}
\newcommand\gappedsucc{\textit{Gapped Succinct Matrix Singularity}}
\newcommand\preciselh{\textit{Precise Local Hamiltonian}}
\newcommand\preciseklh{\textit{Precise }$k$-\textit{Local Hamiltonian}}
\newcommand\preciseilh[1]{\textit{Precise #1-Local Hamiltonian}}

\usepackage{xcolor}

\DeclareMathAlphabet{\matheu}{U}{eus}{m}{n}

\DeclareMathOperator{\tr}{tr}

\newcommand{\poly}{\mathrm{poly}}

\usepackage{mathtools}

\begin{document}

\title{Quantum Merlin Arthur with Exponentially Small Gap}
\author[1]{Bill Fefferman}
\author[1]{Cedric Lin}
\affil[1]{Joint Center for Quantum Information and Computer Science (QuICS), University of Maryland}
\date{\today}

\maketitle

\begin{abstract}
We will study the complexity of $\QMA$ proof systems with inverse exponentially small promise gap.  We will show that this class, $\QMAexp$, can be exactly characterized by $\PSPACE$, the class of problems solvable with a polynomial amount of memory.  As applications we show that a ``precise'' version of the Local Hamiltonian problem is $\PSPACE$-complete, and give a provable setting in which the ability to prepare PEPS states is not as powerful as the ability to prepare the ground state of general Local Hamiltonians.
\end{abstract}

\section{Introduction}
The class $\QMA$, \emph{Quantum Merlin-Arthur}, is the quantum analogue of $\NP$, and is one of the central objects of study in quantum complexity theory. $\QMA$ consists of those problems whose solutions can be verified with high probability using a quantum computer. This class was first shown to have a natural complete problem, the Local Hamiltonian problem, in \cite{ksv02}; since then many more  $\QMA$-complete problems have been discovered (see e.g., \cite{bookatz14}). There has also been much work in trying to prove a quantum version of the PCP theorem; see \cite{qpcpsurvey} for a review.

To be more precise, we give here the definition of $\QMA$:
\begin{definition}We say a promise problem $L=(L_{yes},L_{no})$ is in $\QMA(c,s)$ if there exists a uniform family of quantum circuits $\{ V_x\}_{x\in\{0,1\}^n}$, each of at most polynomial size, and acting on $k(|x|)+m(|x|)$ qubits for polynomials $k$ and $m$, so that:\\

If $x \in L_{yes}$ there exists an $m$-qubit state $\ket{\psi}$ such that:
\begin{equation}
\left(\bra{\psi}\otimes \bra{0^k}\right) V^\dagger_x \ket{1}\bra{1}_{out} V_x \left(\ket{\psi}\otimes \ket{0^k}\right) \ge c
\end{equation}
Whereas if $x \in L_{no}$, for all $m$-qubit states $\ket{\psi}$ we have:
\begin{equation}
\left(\bra{\psi}\otimes \bra{0^k}\right) V^\dagger_x \ket{1}\bra{1}_{out} V_x \left(\ket{\psi}\otimes \ket{0^k}\right) \le s.
\end{equation}
We call $c$ the completeness and $s$ the soundness parameters. Then $\QMA = \QMA(2/3,1/3)$.
  \end{definition}
  
It is natural to wonder whether the precise values of $c$ and $s$ matter. Kitaev showed \cite{ksv02} that as long as $c$ and $s$ are separated by at least an inverse polynomial, then by repeating the verification circuit polynomially many times, it is possible to amplify the promise gap $c-s$ to any constant less than one. Thus $\QMA(c,c-1/\poly) = \QMA$, and quantum Merlin Arthur proof protocols with only a polynomial gap is just as powerful as $\QMA$.

\subsection{Our contribution}

In this work we study the complexity of $\QMA$ protocols where the gap is only exponentially small, i.e. $c-s = \exp(-\poly)$. We show that in this case, the problems verifiable by these protocols exactly coincide with the problems solvable in classical polynomial space:

\begin{restatable}{theorem}{thmmain}
\label{thm:main}
$\QMAexp := \cup_{c-s > \exp(-\poly)}\QMA(c,s) = \PSPACE$.
\end{restatable}

The closest classical counterpart of $\QMAexp$ is $\NP^{\PP}$: given a classical witness, the verifier runs a classical computation that in the YES case accepts with probability at least $c$, or in the NO case accepts with probability at most $s$, where $c>s$. Note that in the classical case the inequality $c - s > \exp(-\poly)$ is always satisfied. Since $\NP^{\PP}$ is in the counting hierarchy, the entirety of which is contained in $\PSPACE$ (see e.g., \cite{allenderwagner}), we see that the quantum proof protocol is strictly stronger than the classical one, unless the counting hierarchy collapses to the second level.

Our proof of this theorem allows us to tweak the proof of $\QMA$-completeness of Local Hamiltonian \cite{ksv02,kr03} to show the following:

\begin{theorem}
\label{thm:lh}
For any $3 \le k \le \mathcal{O}(\log(n))$, determining whether the ground state energy of a $k$-local Hamiltonian is at most $a$ or at least $b$ for $b-a > \exp(-\poly)$ is $\PSPACE$-complete.\footnote{We make no attempt to decide whether $\PSPACE$-completeness still holds for $k = 2$.}
\end{theorem}
In contrast, when $b-a > 1/\poly$ the $k$-Local Hamiltonian problem is $\QMA$-complete \cite{ksv02,kr03} for $2 \le k \le \mathcal{O}(\log(n))$.

Moreover, since we can binary search for the ground state energy in classical polynomial space, this shows
\begin{corollary}
For any $3 \le k \le \mathcal{O}(\log(n))$, computing the ground state energy of a $k$-local Hamiltonian to polynomially many bits of precision is a $\FPSPACE$-complete problem.
\end{corollary}
Here recall that $\FPSPACE$ is the set of functions computable in classical polynomial space.

We have therefore shown that there is, unsurprisingly, a large jump in complexity for the local Hamiltonian problem when the promise gap is only exponentially small instead of polynomially small. Perhaps more surprisingly, $\QMAexp = \PSPACE$ is more powerful than $\PostBQP=\PP$, the class of problems solvable with postselected quantum computation \cite{aaronson05}.

Recall that projected entangled pair states, or PEPS, are a natural extension of matrix product states to two and higher dimensions, and can described as the ground state of certain frustration-free local Hamiltonians \cite{vc04}. A characterization of the computational power of PEPS was given in \cite{swv07}, and can be summarized as follows: let $O_{PEPS}$ be a quantum oracle that, given the description of a PEPS, outputs the PEPS (so the output is quantum). Then $\BQP^{O_{PEPS}}_{\parallel,\text{classical}} = \PostBQP = \PP$, where (following Aaronson \cite{aaronson05}) the subscript  denotes that only classical nonadaptive queries to the oracle are allowed. Moreover, let $\PQP$ stand for the set of problems solvable by a quantum computer with \emph{unbounded error}; then it can be straightforwardly shown that $\PQP^{O_{PEPS}}_{\parallel,\text{classical}} = \PP$ as well (see Appendix \ref{app:peps} for a proof sketch).

On the other hand, suppose we have an oracle $O_{LH}$ that given the description of a local Hamiltonian, outputs its unique ground state\footnote{If the ground space is degenerate, we can always slightly perturb the Hamiltonian to make it nondegenerate.}. Then our results show that $\QMAexp = \PSPACE \subseteq \PQP^{O_{LH}}_{\parallel,\text{classical}}$. This shows that, at least in the setting of unbounded-error quantum computation, PEPS do not capture the full computational complexity of general local Hamiltonian ground states unless $\PP=\PSPACE$. We leave open the problem of determining the complexity of $\BQP^{O_{LH}}_{\parallel,\text{classical}}$.

\section{Definitions}

\subsection{Quantum Merlin Arthur}\label{def:qma}
For our purposes we will need to keep track of the time and space requirements of $\QMA$ protocols, and so we make the following definition:
\begin{definition}We say a promise problem $L=(L_{yes},L_{no})$ is in $\bddQMA{t}{k}{m}{c}{s}$ if there exists a uniform family of quantum circuits $\{ V_x\}_{x\in\{0,1\}^n}$, each of size at most $t(|x|)$, acting on $k(|x|)+m(|x|)$ qubits, so that:\\

If $x \in L_{yes}$ there exists an $m$-qubit state $\ket{\psi}$ such that:
\begin{equation}
\left(\bra{\psi}\otimes \bra{0^k}\right) V^\dagger_x \ket{1}\bra{1}_{out} V_x \left(\ket{\psi}\otimes \ket{0^k}\right) \ge c
\end{equation}
Whereas if $x \in L_{no}$, for all $m$-qubit states $\ket{\psi}$ we have:
\begin{equation}
\left(\bra{\psi}\otimes \bra{0^k}\right) V^\dagger_x \ket{1}\bra{1}_{out} V_x \left(\ket{\psi}\otimes \ket{0^k}\right) \le s.
\end{equation}

  \end{definition}

\begin{remark} $\QMA=\bddQMA{\poly}{\poly}{\poly}{2/3}{1/3}$\end{remark}
\begin{remark} $\QMAexp=\bddQMA{\poly}{\poly}{\poly}{c}{c-2^{-\poly}}$	
\end{remark}

\subsection{Space complexity classes}
\begin{definition}We say a promise problem $L=(L_{yes},L_{no})$ is in $\PQPSPACE$ (unbounded-error quantum polynomial space) if there exists a uniform family of quantum circuits acting on at most polynomial number of qubits that accepts every string $x\in L_{yes}$ with probability at least $c$, and accepts every string $x \in L_{no}$ with probability at most $s$, for some $c > s$.
	
\end{definition}
\begin{theorem}[Watrous \cite{Watrous99,Watrous03}]\label{thm:pqpspace} $\PQPSPACE=\PSPACE$.	
\end{theorem}

\begin{definition}We say a promise problem $L=(L_{yes},L_{no})$ is in $\revPSPACE$ (reversible polynomial space) if it can be decided by a polynomial space reversible Turing machine, i.e., a machine for which every configuration has at most one immediate predecessor.  	
\end{definition}

\begin{theorem}[Bennett \cite{bennett89}]\label{thm:revpspace}$\revPSPACE=\PSPACE$.
\end{theorem}

\section{Upper bound}
In this section, our goal will be to prove that $\QMAexp\subseteq\PSPACE$.  We will proceed in two steps, the first will show how to use in-place $\QMA$ amplification techniques from Nagaj, Wocjan, and Zhang \cite{nwz11} to decide any promise problem in $\QMAexp$ with a quantum protocol in which the verifier is allowed exponential time, polynomial space (i.e., acts on a polynomial number of proof and ancilla qubits), completeness $1-2^{-\poly(n)}$ and soundness $2^{-\poly(n)}$.  We then appeal to results of Marriott and Watrous \cite{mw05} and  Watrous \cite{Watrous99} to show that such protocols can be simulated in $\PSPACE$.
\subsection{In-place gap amplification of $\QMAexp$ using phase estimation techniques}
\begin{theorem}[Implicit in Nagaj, Wocjan, and Zhang \cite{nwz11}] For any $r>0$, \[\bddQMA{t}{k}{m}{c}{s}\subseteq\bddQMA{\mathcal{O}(rt(c-s))}{\mathcal{O}(k+r\log(c-s))}{m}{1-2^{-r}}{2^{-r}}.\]
\end{theorem}
\begin{proof}
	Let $L=(L_{yes}, L_{no})$ be a promise problem in $\QMA(c,s)$ and $\{V_x\}_{x\in\{0,1\}^n}$ the corresponding uniform family of verification circuits.
Define the projectors:
\begin{align}
\Pi_0 &= I_m \otimes \ket{0^k}\bra{0^k} \\
\Pi_1 &= V^\dagger_x \left(\ket{1}\bra{1}_{out} \otimes I_{m+k-1}\right) V_x
\end{align}
and the corresponding reflections:
\begin{align}
R_0 &= 2\Pi_0 - I \\
R_1 &= 2\Pi_1 - I.
\end{align}
Now consider the following procedure:
\begin{enumerate}
\item Perform $r$ trials of phase estimation of the operator $R_1R_0$ on the state $\ket{\psi}\otimes \ket{0^k}$, with $\mathcal{O}(\log(c-s))$ bits of precision and $1/16$ failure probability. 
\item If the median of the $r$ results is at most $\phi_c = \arccos\sqrt{c}/\pi$, output YES; otherwise if the result is at least $\phi_s = \arccos\sqrt{s}/\pi$, output NO.
\end{enumerate}
Phase estimation of an operator $U$ up to $a$ bits of precision requires $\mathcal{O}(a)$ ancilla qubits and $\mathcal{O}(2^a)$ applications of the control-$U$ operation.  Thus, the above procedure, which uses $r$ applications of phase estimation to precision $\alpha=\mathcal{O}(\log(c-s))$ on the $V_x$ operator, can be implemented by a circuit of size $\mathcal{O}(rt2^{\alpha})=\mathcal{O}(rt(c-s))$ using $\mathcal{O}(r\alpha)$ extra ancillia qubits.
 Using the standard analysis of in-place $\QMA$ error amplification \cite{mw05,nwz11}, it can be seen that this procedure has completeness probability at least $1-2^{-r}$ and soundness at most $2^{-r}$.
\end{proof}

Thus, we get the following corollaries:
\begin{corollary}\label{obvious1}For all $r>0$,
$\QMAexp\subseteq\bddQMA{r2^{\poly}}{r\cdot \poly}{\poly}{1-2^{-r}}{2^{-r}}$.
\end{corollary}

\begin{corollary}\label{obvious2}For every problem $L\in\QMAexp$, there exists an $m \in \poly$ so that:
\[L\in\bddQMA{2^{\poly}}{\poly}{m}{1-2^{-(m+2)}}{2^{-(m+2)}}\]
\end{corollary}
Notice that Corollary \ref{obvious2} follows from the definition of $\QMAexp$ and Corollary \ref{obvious1} with $r=m+2$.
\subsection{$\PSPACE$ simulation}
\begin{theorem} \label{thm:pqpspace simulation}
For all $m\in\poly$:\[\bddQMA{2^{\poly}}{\poly}{m}{1-2^{-(m+2)}}{2^{-(m+2)}}\subseteq\PQPSPACE\]\end{theorem}
\begin{proof}
For any $m, p \in\poly$, consider a problem $L\in\bddQMA{2^{\poly}}{p}{m}{1-2^{-(m+2)}}{2^{-(m+2)}}$, and let $\{V'_x\}_{x\in\{0,1\}^n}$ be the corresponding uniform family of verification circuits.
If $x\in L_{yes}$ there exists an $m$-qubit state $\ket{\psi}$ such that
\begin{equation}
\left(\bra{\psi}\otimes \bra{0^p}\right) V'^\dagger_x \ket{1}\bra{1}_{out} V'_x \left(\ket{\psi}\otimes \ket{0^p}\right) \ge 1-2^{-(m+2)}
\end{equation}
whereas if $x \in L_{no}$, for all $m$-qubit states $\ket{\psi}$ we have
\begin{equation}
\left(\bra{\psi}\otimes \bra{0^p}\right) V'^\dagger_x \ket{1}\bra{1}_{out} V'_x \left(\ket{\psi}\otimes \ket{0^p}\right) \le 2^{-(m+2)}.
\end{equation}
For convenience, define the $2^m \times 2^m$ matrix:
\begin{equation}
Q_x := \left(I_{2^m}\otimes \bra{0^p}\right) V'^\dagger_x \ket{1}\bra{1}_{out} V'_x \left(I_{2^m}\otimes \ket{0^p}\right).
\end{equation}
$Q_x$ is positive semidefinite, and $\bra{\psi}Q_x\ket{\psi}$ is the acceptance probability of $V'_x$ on witness $\psi$.
Note that
\begin{equation}
x\in L_{yes} \Rightarrow \tr[Q_x]\ge 1 - 2^{-(m+2)} \ge 3/4
\end{equation}
since the trace is at least the largest eigenvalue, and $m\geq 0$; likewise,
\begin{equation}
x\in L_{no} \Rightarrow \tr[Q_x]\le 2^m \cdot 2^{-(m+2)} = 1/4
\end{equation}
since the trace is the sum of the $2^m$ eigenvalues, each of which is at most $2^{-(m+2)}$. 

Therefore our problem reduces to determining whether the trace of $Q_x$ is at least $3/4$ or at most $1/4$.  Now we will show that using the totally mixed state $2^{-m}I_m$ (alternatively, a random computational basis state) as the witness of the verification procedure encoded by $Q_x$, succeeds with the desired completeness and soundness bounds.  The acceptance probability is given by
\begin{equation}
\tr(Q_x 2^{-m}I_m) = 2^{-m} \tr(Q_x)
\end{equation}
which is at least $2^{-m} \cdot 3/4$ if $x\in L_{yes}$, and at most $2^{-m} \cdot 1/4$ if $x\in L_{no}$. Thus we have reduced our original problem to determining whether an exponentially long quantum computation with \emph{no} witness, acting on a polynomial number of qubits, accepts with probability at least $c'$ or at most $s'$ with $c' - s'$ being exponentially small. This is a $\PQPSPACE$ problem.

\end{proof}
\begin{theorem}
$\QMAexp\subseteq\PSPACE$.
\end{theorem}
\begin{proof}
This follows from Theorem \ref{thm:pqpspace}, Corollary \ref{obvious2}, and Theorem \ref{thm:pqpspace simulation}.
\end{proof}

\section{Lower bound}
In this section we will show that $\PSPACE\subseteq\QMAexp$.  To do this we proceed with two steps.  In the first we show that, given a succinctly representable sparse matrix and promised that either the smallest eigenvalue is 0 or at most $1/2^{\poly}$, deciding which is the case is a $\PSPACE$-complete problem.  In the second step, we give a $\QMAexp$ protocol for this problem.

By a matrix being succinctly representable and sparse, we mean the following:
\begin{definition}
Let $M$ be a $2^{\poly(n)} \times 2^{\poly(n)}$ matrix, where $n$ is the input size. We say that $M$ is a \emph{succinctly representable sparse} matrix if there are at most polynomially many nonzero entries in each row, and moreover there is a (uniformly generated) circuit, \emph{the succinct encoding}, that outputs the nonzero entries of any given row in $\poly(n)$ time.
\end{definition}
\subsection{The Succinct Determinant problem}
\begin{definition}[\succdet]
Given as input is a succinct encoding of $A$, a succinctly representable sparse matrix, whose determinant is promised to be 0, 1, or -1. Moreover, each column of $A$ has at no more than two 1s.  Does $\det(A)$ vanish?
\end{definition}
\begin{theorem}[Grenet, Koiran and Portier \cite{GKP13}]\label{thm:succdet}
$\succdet$ is $\PSPACE$-hard.
\end{theorem}
\begin{proof}
Let $L=(L_{yes},L_{no}) \in \PSPACE$ be decided by a polynomial space deterministic Turing Machine $M$. Consider the configuration graph $G^M$ of $M$: each vertex of $G^M$ corresponds to one of the exponentially many configurations of $M$, each of which is describable with polynomially many bits. The configuration graph $G^M$ has an edge from $c$ to $c'$ if and only if $c'$ can be reached from $c$ in one step of computation. It is straightforward to see that $G^M$ has the following properties:
\begin{itemize}
\item Since $M$ is deterministic, all vertices of $G^M$ have out-degree at most 1, and $G_M$ has no cycles.
\item The adjacency matrix of $G^M$ is a succinctly representable sparse matrix.
\item $M$ accepts input $x$ if and only if there is a path in $G^M$ from the starting configuration $s_x$ to the accepting configuration $t$.
\end{itemize}

Now on input $x$, consider the graph $G^M_x$ obtained by adding an edge from the accepting configuration $t$ to the starting configuration $s_x$, and adding self-loops on all other vertices. 

Recall that a {\sl cycle cover} of a directed graph is a set of disjoint cycles that are subgraphs containing all the vertices of the graph.  We define the {\sl signed weight} of a cycle cover to be the product of the weights of the edges in the cycle cover, multiplied by $(-1)^\ell$, where $\ell$ is the number of cycles of even length in the cycle cover.  We can interpret the determinant of the adjacency matrix of a directed graph as the sum of the signed weights of all cycle covers of the graph. 

Let $A^M_x$ be the adjacency matrix of $G_x^M$.  Now $G^M_x$ has a cycle cover if and only if $M$ accepts $x$ and there is a path from $s_x$ to $t$, and so $\det(A^M_x) = \pm 1$, depending on the signed weight of the cycle cover; otherwise, $\det(A^M_x) = 0$. Therefore deciding whether $\det(A^M_x)$ vanishes is $\PSPACE$-hard.	
\end{proof}
We can immediately see that the complexity of the $\succdet$ problem doesn't get easier if we are promised that the succinctly representable sparse input matrix $A$ is symmetric and positive semidefinite: notice that the  matrix $(A^M_x)^T A^M_x$ is succinctly representable sparse because there are at most two $1$'s in each column of $A^M_x$ and if we can decide if $\det((A^M_x)^T A^M_x)=\det(A^M_x)^2$ vanishes (or is equal to $1$), we can certainly decide if $\det(A^M_x)$ vanishes. 

From here, we would like to argue that given a succinctly representable sparse and symmetric PSD matrix $A$, it is $\PSPACE$-hard to determine whether the smallest eigenvalue $\lambda_{min}$ satisfies $\lambda_{min} = 0$ or $\lambda_{min} > 2^{-\text{poly}}$, promised that one of these is the case. We will see later that this problem can be solved in $\QMAexp$. Unfortunately, this promise does not generally hold for succinctly representable sparse symmetric matrices; if $A$ is nonsingular, the smallest eigenvalue can at worst still be doubly exponentially small. We will therefore need to modify the prior $\PSPACE$-hard construction to show that the following $\gappedsucc$ problem is still $\PSPACE$-hard.
\begin{definition}[$\gappedsucc$]
Given as input is a succinct encoding of $A$, a positive semidefinite, symmetric, and succinctly representable sparse matrix, whose entries are 0, 1, or 2. Moreover, the smallest eigenvalue of $A$ is promised to be either zero or at least $2^{-g(n)}$ for some polynomial $g(n)$. Output YES if the smallest eigenvalue of $A$ is zero.
\end{definition}
\begin{theorem}\label{thm:gappedsucc}
$\gappedsucc$ is $\PSPACE$-hard.
\end{theorem}

\begin{proof}
	In this theorem, we will adapt the construction of Theorem \ref{thm:succdet} to analyze the spectrum of the underlying $\PSPACE$ machine configuration graph.  
	 Recall we defined $\revPSPACE$ to be the class of languages decidable in polynomial space by a reversible Turing Machine.  Theorem \ref{thm:revpspace} states that $\PSPACE=\revPSPACE$, and indeed a result of Lange, McKenzie and Tapp \cite{lange00} proves that for any space-constructible $s$, ${\sf{DSPACE}[s]}\subseteq{\sf{revSPACE}[s]}$, at a cost of an exponential time blow-up.  In fact, they also show without loss of generality that the starting configuration of the reversible machine has in-degree $0$, as long as the input $x$ is kept on the tape at the end of the computation (and so the accepting configuration depends on $x$).
	
	Thus, an arbitrary $\PSPACE$ language $L$ can be decided by a polynomial space reversible Turing machine $M$, and the resulting configuration graph $G^M$ is a collection of disjoint paths. 
		
As before consider the graph $G^M_x$ obtained by adding an edge from the accepting configuration $t$ to the starting configuration $s_x$, and adding self-loops on all other vertices. 
Now, if $\det(A^M_x) \neq 0$, i.e. if $M$ accepts $x$, there is a maximal path in $G^M$ starting from $s_x$ and ending at $t_x$. Assume the path has $\ell+1$ vertices, where $\ell \in 2^{\poly}$. $G^M_x$ adds an edge from $t_x$ to $s_x$ and adds a self-loop to all other vertices. Therefore if $M$ accepts $x$, $G^M_x$ is a disjoint union of connected graphs $G^M_{x,i}$, where:
\begin{enumerate}
\item If $s_x, t_x$ are not vertices of $G^M_{x,i}$, $G^M_{x,i}$ is a path with additional self-loops on all vertices of the path.
\item If $s_x, t_x$ are vertices of $G^M_{x,i}$, $G^M_{x,i}$ is a cycle, with $s_x$ coming directly after $t_x$ in the cycle, and with additional self-loops on all vertices in the cycle except for $s_x$ and $t_x$.
\end{enumerate}
Let us look at these two cases separately. Assume a subgraph of the first type has $\ell$ vertices (i.e. the path has length $\ell-1$); then its adjacency matrix is, after appropriate relabelling of vertices, the following $\ell \times \ell$ matrix:
\begin{equation}
A_{1,\ell} := 
\begin{bmatrix}
    1 & 0 & 0 & 0 & \dots  & 0  & 0 \\
    1 & 1 & 0 & 0 &\dots  & 0 & 0 \\
    0 & 1 & 1 & 0 & \dots  & 0 & 0 \\
     0 & 0 & 1 & 1 & \dots  & 0 & 0 \\
    \vdots & \vdots & \vdots & \vdots & \ddots & \vdots & \vdots \\
    0 & 0 & 0 & 0 & \dots  & 1 & 0 \\
    0 & 0 & 0 & 0 & \dots  & 1 & 1
\end{bmatrix}.
\end{equation}
Computing $A_{1,\ell}^T A_{1,\ell}$, we see that it is:
\begin{equation} \label{eq:psd_mat}
A_{1,\ell}^T A_{1,\ell} := 
\begin{bmatrix}
    2 & 1 & 0 & 0 & \dots  & 0  & 0 \\
    1 & 2 & 1 & 0 &\dots  & 0 & 0 \\
    0 & 1 & 2 & 1 & \dots  & 0 & 0 \\
     0 & 0 & 1 & 2 & \dots  & 0 & 0 \\
    \vdots & \vdots & \vdots & \vdots & \ddots & \vdots & \vdots \\
    0 & 0 & 0 & 0 & \dots  & 2 & 1 \\
    0 & 0 & 0 & 0 & \dots  & 1 & 1
\end{bmatrix}.
\end{equation}
The characteristic equation for the eigenvalues $\lambda$ is then $p_\ell(\lambda) = 0$, where
\begin{equation} \label{eq:char_p}
p_\ell(\lambda) := 
\det \left(
\begin{bmatrix}
    2 - \lambda & 1 & 0 & 0 & \dots  & 0 & 0  \\
    1 & 2 - \lambda & 1 & 0 &\dots  & 0 & 0 \\
    0 & 1 & 2 - \lambda & 1 & \dots  & 0 & 0 \\
     0 & 0 & 1 & 2 - \lambda & \dots  & 0 & 0 \\
    \vdots & \vdots & \vdots & \vdots & \ddots & \vdots \\
    0 & 0 & 0 & 0 & \dots  & 2 - \lambda & 1 \\
    0 & 0 & 0 & 0 & \dots  & 1 & 1 - \lambda \\
\end{bmatrix} \right).
\end{equation}
Now define the polynomial $q_n(x)$ to be the following $n \times n$ determinant:
\begin{equation}
q_n(x) := 
\det \left(
\begin{bmatrix}
    x & 1 & 0 & \dots  & 0 & 0  \\
    1 & x & 1 &\dots  & 0 & 0 \\
    0 & 1 & x  & \dots  & 0 & 0 \\
    \vdots & \vdots & \vdots & \ddots & \vdots \\
    0 & 0 & 0 & \dots  & x & 1 \\
    0 & 0 & 0 & \dots  & 1 & x \\
\end{bmatrix} \right)
\end{equation}
Note that $p_\ell(\lambda) = q_\ell(2-\lambda) - q_{\ell-1}(2-\lambda)$. Moreover, $q_n(x)$ satisfies the recurrence
\begin{equation}
q_0(x) = 1, \quad q_1(x) = x, \quad q_n(x) = x q_{n-1}(x) - q_{n-2} (x)
\end{equation}
This is the same recurrence satisfied by $U_n(x/2)$, where $U_n(x)$ is the $n$-th degree Chebyshev polynomial of the second kind. Therefore $q_n(x) = U_n(x/2)$. Using $U_n(\cos\theta) = \sin((n+1)\theta)/\sin(\theta)$ we can evaluate:
\begin{align}
q_\ell(2\cos\theta) - q_{\ell-1}(2\cos\theta) &= \frac{\sin((\ell+1)\theta) - \sin(\ell \theta)}{\sin\theta} \\
&= \frac{\sin((2\ell+1)\theta/2)}{\sin (\theta/2)}
\end{align}
and therefore the zeroes of $q_\ell(x) - q_{\ell-1}(x)$ are $2\cos\left(\frac{2k}{2\ell + 1}\pi\right)$, $k = 1,\cdots,\ell$. The zeroes of the polynomial $p_\ell(\lambda) = q_\ell(2-\lambda) - q_{\ell-1}(2-\lambda)$ are then
\begin{equation}
\lambda_k = 2\left(1 - \cos\left(\frac{2k}{2\ell+1}\pi\right)\right)
\end{equation}
and the smallest eigenvalue $\lambda_1 = \Theta(\ell^{-2})$ is inverse exponentially bounded away from zero, because $\ell = 2^{\mathcal{O}(\poly)}$.

We now look at the other case, where the subgraph is of the second type, i.e. it contains $s_x$ and $t_x$. The adjacency matrix of this subgraph is, assuming there are $\ell$ vertices, the $\ell \times \ell $ matrix:
\begin{equation}
A_{2,\ell} := 
\begin{bmatrix}
    0 & 0 & 0 & 0 & \dots  & 0  & 1 \\
    1 & 1 & 0 & 0 &\dots  & 0 & 0 \\
    0 & 1 & 1 & 0 & \dots  & 0 & 0 \\
     0 & 0 & 1 & 1 & \dots  & 0 & 0 \\
    \vdots & \vdots & \vdots & \vdots & \ddots & \vdots & \vdots \\
    0 & 0 & 0 & 0 & \dots  & 1 & 0 \\
    0 & 0 & 0 & 0 & \dots  & 1 & 0
\end{bmatrix}
\end{equation}
We can directly evaluate $A_{2,\ell}^T A_{2,\ell}$, obtaining the following matrix:
\begin{equation}
A_{2,\ell}^T A_{2,\ell} = 
\begin{bmatrix}
    1 & 1 & 0 & 0 & \dots  & 0 & 0  & 0 \\
    1 & 2 & 1 & 0 &\dots  & 0 & 0 & 0 \\
    0 & 1 & 2 & 1 & \dots  & 0 & 0 & 0 \\
     0 & 0 & 1 & 2 & \dots  & 0 & 0 & 0 \\
    \vdots & \vdots & \vdots & \vdots & \ddots & \vdots & \vdots \\
    0 & 0 & 0 & 0 & \dots  & 2 & 1 & 0 \\
    0 & 0 & 0 & 0 & \dots  & 1 & 2 & 0 \\
    0 & 0 & 0 & 0 & \dots  & 0 & 0 & 1
\end{bmatrix}
\end{equation}
For purposes of calculating the smallest eigenvalue the last row and column can be ignored, leaving an $(\ell-1) \times (\ell-1)$ matrix. Looking back at \ref{eq:psd_mat} and \ref{eq:char_p}, we see that the characteristic equation is exactly $p_{\ell-1}(\lambda) = 0$. Therefore once again the smallest eigenvalue is inverse exponentially bounded away from zero.

\end{proof}
\subsection{$\gappedsucc$ is in $\QMAexp$}
In this section, we now proceed to show a $\QMAexp$ protocol for $\gappedsucc$.
\begin{lemma} \label{lem:qmaexp protocol}
Let $A$ be a positive semidefinite, symmetric, and succinctly representable sparse matrix, whose entries are 0, 1, or 2; moreover the smallest eigenvalue of $A$ is promised to be either zero or at least $2^{-g(n)}$ for some polynomial $g(n)$. There is a $\QMAexp$ protocol for deciding which is the case.
\end{lemma}
Our strategy will essentially be to simulate the time evolution of the sparse Hamiltonian $e^{-iAt}$ using known simulation methods, and then use a stripped-down version of phase estimation to estimate an eigenvalue of $A$. We first note the following result for sparse matrix simulation:
\begin{theorem}[\cite{berry14}, \cite{berry15}] \label{thm:ham_sim}
Suppose $A$ is a $2^n \times 2^n$ symmetric and succinctly representable sparse matrix, with at most $d$ nonzero entries in each row. Then treated as a Hamiltonian, the time evolution $\exp(-iAt)$ can be simulated using $\text{\emph{poly}}(n,d,\|A\|, t, \log(1/\epsilon))$ operations.
\end{theorem}
The crucial thing to notice in Theorem \ref{thm:ham_sim} is the polylogarithmic scaling in the error $\epsilon$; this implies that we can obtain exponential precision in $\exp(-iAt)$ using only polynomially many operations. Also note that we can upper bound $\|A\|$ with the following observation:
\begin{remark}
Suppose a matrix $A$ has at most $d$ nonzero entries per row, each of which is no more than $k$ in absolute value. Then $\| A \| \le kd$.
\end{remark}
\begin{proof}[Proof of Lemma \ref{lem:qmaexp protocol}]
We are given a succinct encoding of an symmetric PSD $d$-sparse matrix $A$, and it is promised that the smallest eigenvalue $\lambda_{min}$ of $A$ is either zero or at least $2^{-g(n)}$ for some polynomial $g(n)$. Merlin would like to convince us that $\lambda_{min} = 0$; he will send us a purported eigenstate $\ket{\psi}$ of $A$ with zero eigenvalue. We will carry out a stripped-down version of phase estimation on $\exp(-iAt)$ acting on $\ket{\psi}$ to decide, with exponentially small completeness-soundness gap, whether Merlin is telling the truth. Let us choose $t = \pi / (kd) \le \pi / \|A\|$; then all eigenvalues of $At$ lie in the range $[0,\pi]$, and the output of phase estimation will be unambiguous.

To implement phase estimation, we first need to be able to implement $\exp(-iAt)$ efficiently and to high precision. This is what Theorem \ref{thm:ham_sim} gives us: we can implement $\exp(-iAt)$ up to error $\epsilon = 2^{-\poly(n)}$ using only a polynomial number of operations, for any choice of $\epsilon$.

Now to use phase estimation to distinguish the phase up to exponential precision, we would normally require exponentially many operations in the usual phase estimation routine. Instead, we will simply do phase estimation with one bit:
\begin{align}
&\Qcircuit @C=1em @R=.7em {
\lstick{\ket{0}}& \gate{H} & \ctrl{1} & \gate{H} & \rstick{\frac{1+e^{-i\lambda t}}{2}\ket{0} +\frac{1-e^{-i\lambda t}}{2}\ket{1} } \qw \\
\lstick{\ket{\psi}}& \qw & \gate{e^{-iAt}}  & \qw & \rstick{\ket{\psi}} \qw
}
\end{align}
In the above we've assumed $\ket{\psi}$ is an eigenstate of $A$ with eigenvalue $\lambda$. If we measure the control qubit at the end, we see the probability we obtain 0 is $1 - (1-\cos(\lambda t))/2 = 1 - (\lambda t)^2/4 + \mathcal{O}(\lambda^4t^4)$. Therefore if $\psi$ is a zero eigenstate, we can verify this with probability at least $1 - \epsilon$, where recall $\epsilon$ is the error in the implementation of $\exp(-iAt)$. Otherwise if $\lambda_{min} \ge 2^{-g(n)}$, no state $\psi$ will be accepted with probability greater than $1 -  2^{-2g(n)}t^2/4 + \epsilon + \mathcal{O}(2^{-4g(n)}t^4)$. The separation between the completeness and soundness probabilities is exponentially small if we pick $\epsilon \le 2^{-2g(n)}t^2/16$, and this therefore gives us a $\QMAexp$ protocol.
\end{proof}
\begin{theorem}
$\PSPACE\subseteq\QMAexp$.
\end{theorem}
\begin{proof}
Follows from Theorem \ref{thm:gappedsucc} and Lemma \ref{lem:qmaexp protocol}.
\end{proof}
This finishes the proof of our main theorem:
\thmmain*

\section{A $\PSPACE$-complete variant of the Local Hamiltonian problem}
The classic $\QMA$-complete problem is the Local Hamiltonian problem: given a local Hamiltonian $H$, and parameters $a < b$ with $b-a > 1/\poly$, it is promised that the smallest eigenvalue of $H$ is either at most $a$ or at least $b$; decide which is the case. We now show that if we weaken the promise gap from polynomially small to only exponentially small, then this problem becomes $\PSPACE$-complete.
\begin{definition}[\preciseklh]
Given as input is a $k$-local Hamiltonian $H=\sum_{j=1}^rH_j$ acting on $n$ qubits, satisfying $r \in \poly(n)$ and $\|H_j\| \le \poly(n)$, and numbers $a < b$ satisfying $b - a > 2^{-\poly(n)}$. It is promised that the smallest eigenvalue of $H$ is either at most $a$ or at least $b$. Output 1 if the smallest eigenvalue of $H$ is at most $a$, and output 0 otherwise.
\end{definition}

\begingroup
\def\thetheorem{\ref{thm:lh}}
\begin{theorem}
For any $3 \le k \le \mathcal{O}(\log(n))$, \preciseklh \ is $\QMAexp$-complete, and hence $\PSPACE$-complete.
\end{theorem}
\addtocounter{theorem}{-1}
\endgroup

\begin{proof}
This proof follows straightforwardly by adapting the proof of \cite{ksv02} and \cite{kr03}. The proof of containment in $\QMAexp$ is identical to the containment of the usual Local Hamiltonian problem in $\QMA$; see \cite{ksv02} for details.

To show $\QMAexp$-hardness, we note that for a $\QMA$-verification procedure with $T$ gates, completeness $c$ and soundness $s$, \cite{kr03} reduces this to a 3-local Hamiltonian with lowest eigenvalue no more than $(1-c) / (T+1)$ in the YES case, or no less than $(1-s) / T^3$ in the NO case. For this to specify a valid \preciselh \ problem we need that
\begin{equation} \label{eq:preciselh_condition}
\frac{1-s}{T^3} - \frac{1-c}{T+1} > 2^{-\poly(n)}.
\end{equation}
Fortunately, there are indeed values of $c$ and $s$ that satisfy the above inequality and can still specify $\QMAexp$-hard problems. To see this, we recall the proof of Lemma \ref{lem:qmaexp protocol}: there it was shown that any problem in $\PSPACE$ can be reduced to a $\QMAexp$ problem with soundness and completeness
\begin{equation}
1-c = \epsilon,\quad 1 - s = -\epsilon + 2^{-g'(n)}
\end{equation}
for some polynomial $g'(n)$ depending on the problem, and any $\epsilon = 2^{-\poly(n)}$ of our choice. The number of operations for that protocol is upper bounded by $T \le h(n,\log (1/\epsilon))$ for some polynomial $h(x,y)$. Now we can pick $\epsilon$ to be a small enough inverse exponential function such that
\begin{equation}
\epsilon (T^2+1) \le \epsilon [h^2(n,\log(1/\epsilon)) +1] < 2^{-g'(n)} 
\end{equation}
holds; this then implies the inequality \ref{eq:preciselh_condition}. Hence any problem in $\PSPACE$ can be reduced to a \preciseilh{3} problem.
\end{proof}
\section{Acknowledgements}
We are grateful to Sevag Gharibian and Martin Schwarz for helpful conversations, and to John Watrous for comments on a preliminary draft. This work was supported by the Department of Defense.
\bibliography{qmaexp-bib}
\bibliographystyle{plain}

\appendix

\section{Proof sketch of $\PQP^{O_{PEPS}}_{\parallel,\text{classical}} = \PP$} \label{app:peps}
Since $\PP \subseteq \BQP^{O_{PEPS}}_{\parallel,\text{classical}} \subseteq \PQP^{O_{PEPS}}_{\parallel,\text{classical}}$ \cite{swv07}, we only need to show that $\PQP^{O_{PEPS}}_{\parallel,\text{classical}} \subseteq \PP$. In \cite{swv07} it was noted that all PEPS can be seen as the output of a quantum circuit followed by a postselected measurement. Therefore $\PQP^{O_{PEPS}}_{\parallel,\text{classical}}$ corresponds to the problems that can be decided by a quantum circuit, followed by a postselected measurement (since the queries to $O_{PEPS}$ are classical and nonadaptive, we can compose them into one single postselection), followed by a measurement. In the YES case the measurement outputs 1 with probability at least $c$, whereas in the NO case the measurement outputs 1 with probability at most $s$, with $c > s$. The standard counting argument placing $\BQP$ inside $\PP$ then applies to this case as well; see for instance \cite[Propositions~2~and~3]{aaronson05}.

\end{document}